\documentclass[twocolumn,superscriptaddress,
showpacs,preprintnumbers,amsmath,amssymb]{revtex4}
\usepackage{amssymb,amsmath,amsthm}
\usepackage{epsfig}
\newtheorem{Thm}{Theorem}
\newtheorem{Cor}{Corollary}
\theoremstyle{definition}
\newtheorem{counterexample}{Counterexample}

\newcommand{\bra}[1]{{\left\langle #1 \right|}}
\newcommand{\ket}[1]{{\left| #1 \right\rangle}}

\newcommand{\T}{\mbox{$\mathrm{tr}$}}

\begin{document}
%%%%%%%%%%%%%%%%%%%%%%%%%%%%%%%%%%%%%%%%%%%%%%%%%%%%%%%%%%%%%%%%%%%%%%%%%%
%                                                                        %
%                                 Title                                  %
%                                                                        %
%%%%%%%%%%%%%%%%%%%%%%%%%%%%%%%%%%%%%%%%%%%%%%%%%%%%%%%%%%%%%%%%%%%%%%%%%%
\title{Entanglement monogamy of multipartite higher-dimensional quantum systems using convex-roof extended negativity}

\author{Jeong San Kim}
%\email{jkim@qis.ucalgary.ca}
\affiliation{
 Institute for Quantum Information Science,
 University of Calgary, Alberta T2N 1N4, Canada
}
\author{Anirban Das}
%\email{anirband@usc.edu}
\affiliation{
 Institute for Quantum Information Science, University of Calgary, Alberta T2N 1N4, Canada
}
\affiliation{
 Department of Physics and Astronomy, University of Southern California, Los Angeles, CA 90089, USA
}
\author{Barry C. Sanders}
%\email{bsanders@qis.ucalgary.ca}
\affiliation{
 Institute for Quantum Information Science,
 University of Calgary, Alberta T2N 1N4, Canada
}

\date{\today}

%%%%%%%%%%%%%%%%%%%%%%%%%%%%%%%%%%%%%%%%%%%%%%%%%%%%%%%%%%%%%%%%%%%%%%%%%%
%                                                                        %
%                              Abstract                                  %
%                                                                        %
%%%%%%%%%%%%%%%%%%%%%%%%%%%%%%%%%%%%%%%%%%%%%%%%%%%%%%%%%%%%%%%%%%%%%%%%%%
\begin{abstract}
We propose replacing concurrence by convex-roof extended negativity (CREN)
for studying monogamy of entanglement (MoE).
We show that all proven MoE relations using concurrence can be rephrased in terms of CREN.
Furthermore we show that higher-dimensional (qudit) extensions of MoE in terms
of CREN are not disproven by any of the counterexamples used to disprove qudit
extensions of MoE in terms of concurrence. We further test the CREN version of MoE
for qudits by considering fully or partially coherent  mixtures of a qudit W-class state
with the vacuum and show that the CREN version of MoE for qudits is satisfied
in this case as well. The CREN version of MoE for qudits is thus a strong
conjecture with no obvious counterexamples.
\end{abstract}

%\keywords{}
\maketitle

%%%%%%%%%%%%%%%%%%%%%%%%%%%%%%%%%%%%%%%%%%%%%%%%%%%%%%%%%%%%%%%%%%%%%
%%                                                                %%%
%%                         Introduction                           %%%
%%                                                                %%%
%%%%%%%%%%%%%%%%%%%%%%%%%%%%%%%%%%%%%%%%%%%%%%%%%%%%%%%%%%%%%%%%%%%%%
\section{Introduction}
Quantum entanglement is a resource with various applications such
as quantum teleportation and quantum key distribution in the field
of quantum information and quantum computation~\cite{tele, qkd1, qkd2}.
Whereas entanglement in bipartite quantum systems has been
intensively studied with rich understanding, the situation becomes
far more difficult for the case of multipartite quantum systems, and
very few are known for its characterization and quantification.
One important property to characterize multipartite entanglement is
known as {\em monogamy of entanglement} (MoE)~\cite{ckw}, which says that
multipartite entanglements cannot be freely shared among the parties.

MoE is a key ingredient to make quantum
cryptography secure~\cite{rg}, and it also plays an important role
%For example, if two distant parties are
%maximally entangled with each other, then they can be assured
%their relation is secure against any other party including a
%potential eavesdropper due to the monogamy property of
%entanglement.
in condensed-matter physics such as the frustration effects observed
in Heisenberg antiferromagnets and the $N$-representability problem for
fermions~\cite{anti, d, h}.
%Because a perfect ground state for an antiferromagnet tends to
%distribute its entanglements of its neighboring spins to be optimal,
%
%%Because a perfect ground state for
%%an antiferromagnet tends to be a singlet between all interacting
%%spins, a particle tries to distribute its entanglements with its
%%neighbors to be optimal, that is, a strongly correlated ground
%%state.
%
%this behavior can have a clear mathematical interpretation with quantitative statements in terms
%of MoE. Similarly, the $N$-representability problem for
%fermions~\cite{anti} also can have an analogous interpretation of
%quantifying the distributed entanglement in multipartite systems
%by means of MoE.
Thus, it is an important and necessary task to
characterize MoE to understand the whole picture of quantum entanglement
in multipartite systems, as well as its possible applications in
quantum information theory.

Although MoE is a typical property of multipartite quantum entanglement,
it is however about the relation of bipartite entanglements among the parties in multipartite systems.
In other words, it is inevitable and even crucial to have a proper way of quantifying bipartite
entanglement for a good description of the monogamy nature in multipartite quantum systems.
Thus, the following criteria must be satisfied for a proper choice of an entanglement measure.
\begin{itemize}
\item[(i)] {\em Monotonicity}: describes the non-local
character of quantum entanglement, that is, the amount of
entanglement is not increased under LOCC.

\item[(ii)] {\em Separability}: capability
of distinguishing entanglement from separability.
\end{itemize}

\begin{itemize}
\item[(iii)] {\em Monogamy}: upper bound on a sum of bipartite entanglement
measures thereby showing that bipartite sharing of entanglement is bounded.
\end{itemize}

There are several possibilities for such a measure, including
definitively answering whether the state is entangled or
separable, indicating definitively that the state is entangled but
inconclusive when the result is `separable' as well as the reverse
case, and stating whether the state is entangled and/or separable
with bounded error.

However, there are only a few measures known so far which can show the monogamy property
of entanglement in multipartite systems, and their results are restricted to
multi-qubit systems~\cite{ckw, ov}.
In other words, there exist quantum states in higher-dimensional systems~\cite{ou, ks}
which violate the monogamy properties in terms of the proposed entanglement measures,
and this exposes the importance of choosing a proper entanglement measure.

Here we propose the {\em convex-roof extended negativity}
(CREN)~\cite{LCOK} as a powerful candidate for the criteria above.
Besides its monotonicity and separability criteria, we claim that CREN is a
good alternative for MoE without any known example violating its monogamy property
even in higher-dimensional systems.
We show that any monogamy inequality of entanglement for multi-qubit systems
using concurrence~\cite{ww} can be rephrased by CREN, and this CREN MoE is
also true for the counterexamples of concurrence in higher-dimensional systems~\cite{ou, ks}.

As the first step toward general CREN MoE studies in higher-dimensional quantum systems,
we propose a class of quantum states in $n$-qudit systems consisting of partially
coherent superpositions of a generalized W-class state~\cite{ks} and the
vacuum, $\ket{0}^{\otimes n}$, and show that this class
saturates CREN MoE for any arbitrary partition of the set of subsystems.
We also show that the CREN value of the proposed class and its dual, {\em CREN of Assistance} (CRENoA)
coincide, and they are not affected by the degree of coherency in the superposition.
This is particularly important because the saturation of monogamy
relation implies that this class of multipartite higher-dimensional entanglement
can have a complete characterization by means of its partial entanglements,
and the characterization is not even affected by its decoherency.

The paper is organized as follows. In Sec.~\ref{Sec:
Concurrence and CREN}, we reprise the definitions of concurrence,
CREN, and their overlap for the case of pure states with
Schmidt rank 2, as well as two-qubit mixed states. In
Sec.~\ref{Subsec: Monogamy inequalities for $n$-qubit systems
in terms of CREN}, we rephrase all the monogamy inequalities of
entanglement for $n$-qubit systems in terms of CREN.
In Sec.~\ref{Subsec: CREN Monogamy Relations in the
Counterexamples of CKW inequality}, we show
that the counterexamples in higher-dimensional
quantum systems to the monogamy inequality using concurrence
still have a monogamy relation in terms of CREN.
In Sec.~\ref{Sec: Partially Coherent Superposition}, a class
of quantum states in $n$-qudit systems consisting of partially
coherent superpositions of a generalized W-class state and $\ket{0}^{\otimes n}$
is proposed with its CREN monogamy relation of entanglement.
In Sec.~\ref{Sec: Conclusion}, we summarize our results.

%%%%%%%%%%%%%%%%%%%%%%%%%%%%%%%%%%%%%%%%%%%%%%%%%%%%%%%%%%%%%%%%%%%%%
%%                                                                %%%
%%       Concurrence and Convex-Roof Extended Negativity          %%%
%%                                                                %%%
%%%%%%%%%%%%%%%%%%%%%%%%%%%%%%%%%%%%%%%%%%%%%%%%%%%%%%%%%%%%%%%%%%%%%

\section{Concurrence and Convex-Roof Extended Negativity}
\label{Sec: Concurrence and CREN}

For any bipartite pure state $\ket \phi_{AB}$ in a $d\otimes d'$
($d\le d'$) quantum system, its {\em concurrence}, $\mathcal{C}(\ket
\phi_{AB})$ is defined as~\cite{ww}
\begin{equation}
\mathcal{C}(\ket \phi_{AB})=\sqrt{2(1-\T\rho^2_A)},
\label{pure state concurrence}
\end{equation}
where $\rho_A=\T_B(\ket \phi_{AB}\bra \phi)$.
For any mixed state $\rho_{AB}$, it is defined as
\begin{equation}
\mathcal{C}(\rho_{AB})=\min \sum_k p_k \mathcal{C}({\ket {\phi_k}}_{AB}),
\label{mixed state concurrence}
\end{equation}
where the minimum is take over all possible pure state
decompositions, $\rho_{AB}=\sum_kp_k{\ket {\phi_k}}_{AB}\bra
\phi_k$.

{\em Concurrence of Assistance} (CoA)~\cite{lve}, which can be considered to be dual to
concurrence, is defined as
\begin{equation}
\mathcal{C}^a(\rho_{AB})=\max \sum_k p_k \mathcal{C}({\ket {\phi_k}}_{AB}),
\label{CoA}
\end{equation}
where the maximum is taken over all possible pure state decompositions of
$\rho_{AB}$.

Another well-known quantification of bipartite entanglement is
the {\em negativity}~\cite{VidalW, LCOK}, which is based on
the {\em positive partial transposition} (PPT) criterion~\cite{Peres,
Horodeckis1}.
For a bipartite pure state $\ket{\phi}_{AB}$ in a $d\otimes d'$
($d\le d'$) quantum system with the Schmidt decomposition,
\begin{equation}
\ket{\phi}_{AB}~=~\sum_{i=0}^{d-1} \sqrt{\lambda_{i}}\ket{ii},~~\lambda_{i} \geq 0,~\sum_{i=0}^{d-1}\lambda_{i}~=1,
\label{Schmidt}
\end{equation}
(without loss of generality, the Schmidt basis is taken to be the
standard basis), the partial transposition of $\ket{\phi}_{AB}$ is
\begin{align}
\ket{\phi}\bra{\phi}^{T_B} =& \sum_{i,j=0}^{d-1}\sqrt{\lambda_{i}\lambda_{j}}\ket{ij}\bra{ji}\nonumber\\
%&=& \sum_{i=1}^{d} a_{i}^{2}\ket{ii}\bra{ii} + \sum_{i\neq
%j}a_{i}a_{j}\ket{ij}\bra{ji}\\
=& \sum_{i=0}^{d-1} \lambda_{i}\ket{ii}\bra{ii} +
\sum_{i<j} \sqrt{\lambda_{i}\lambda_{j}}(\ket{ij}\bra{ji}+\ket{ji}\bra{ij}).
\label{eq:pt state}
\end{align}
Thus, the negative eigenvalues can be $- \lambda_{i}\lambda_{j}$
for $i < j$ with corresponding eigenvectors
$\ket{\psi_{ij}}~=~\frac{1}{\sqrt{2}}(\ket{ij}-\ket{ji})$, and the
negativity $\mathcal{N}$ of $\ket{\phi}_{AB}$ is defined
as~\cite{Negativity}
\begin{align}
\mathcal{N}(\ket{\phi})&=\left\|\ket{\phi}\bra{\phi}^{T_B}\right\|_1-1 \nonumber\\
&=2\sum_{i<j}\sqrt{\lambda_{i}\lambda_{j}},
\label{eq:pure_negativity}
\end{align}
where $\left\|\cdot\right\|_1$ is the trace norm.

Based on the reduced density matrix of $\ket{\phi}_{AB}$,
we can have an alternative definition of negativity,
\begin{align}
\mathcal{N}(\ket{\phi})&=2\sum_{i<j}\sqrt{\lambda_{i}\lambda_{j}}\nonumber\\
&=(\T{ \sqrt{\rho_A}})^2 -1,
\label{eq:pure_negativity2}
\end{align}
where $\rho_A = \T_B{\ket{\phi}_{AB}\bra{\phi}}$.

We note that $\mathcal{N}(\ket{\phi})=0$ if and only if
$\ket{\psi}$ is separable, and it can attain its maximal value,
$d-1$, for a $d\otimes d$ maximally entangled state,
\begin{equation}
\ket{\phi}=\frac{1}{\sqrt d}\sum_{i=0}^{d-1}\ket{ii}.
\label{edit}
\end{equation}
(One can easily check this by the {\em Lagrange multiplier}.)
%Thus $\mathcal{N}$ can be a measure of entanglement
%for bipartite pure states in any dimensional quantum system.

For a mixed state $\rho_{AB}$, its negativity is defined as
\begin{equation}
\mathcal{N}(\rho_{AB})=\left\|{\rho_{AB}}^{T_B}\right\|_1-1,\label{eq:negativity}
\end{equation}
where $\rho^{T_B}$ is the partial transpose of $\rho_{AB}$.

It is known that PPT gives a separability criterion for two-qubit systems, and it is
also a necessary and sufficient condition for nondistillability in
$2\otimes n$ quantum system \cite{Horodecki1,DCLB}. However, in
higher-dimensional quantum systems rather than $2\otimes 2$ and
$2\otimes 3$ quantum systems, there exist mixed entangled states with
PPT, so-called `bound entangled
states'~\cite{Horodecki1,Horodeckis2}. For this case, negativity
cannot distinguish PPT bound entangled states from separable
states, and thus, negativity itself is not sufficient to be a good
measure of entanglement even in a $2\otimes n$ quantum system.

One modification of negativity to overcome its lack of
separability criterion is CREN~\cite{unnomalization}, which gives a perfect
discrimination of PPT bound entangled states and separable states
in any bipartite quantum system.

For a bipartite mixed state mixed state $\rho_{AB}$, CREN is defined as
\begin{equation}
\mathcal{N}_c(\rho) \equiv\min \sum_k p_k
\mathcal{N} (\ket{\phi}_k), \label{eq:c_negativity}
\end{equation}
where the minimum is taken over all possible pure state
decompositions of $\rho={\sum_k p_k \ket{\phi_k}\bra{\phi_k}}$.

Whereas a normalized version of the negativity depending on the dimension of the
quantum systems was used to show its monotonicity~\cite{LCOK}, it can be
analogously shown with the definitions in Eqs.~(\ref{eq:negativity}) and (\ref{eq:c_negativity}).

Now, let us consider the relation between CREN and concurrence.
For any bipartite pure state $\ket{\phi}_{AB}$ in a $d \otimes d'$
quantum system with Schmidt rank 2,
\begin{equation}
\ket{\phi}=\sqrt{\lambda_0}\ket{00}+\sqrt{\lambda_1}\ket{11},
\label{schmidt2}
\end{equation}
we have
\begin{align}
\mathcal{N}(\ket{\phi})&=\left\|\ket{\phi}\bra{\phi}^{T_B}\right\|_1-1 \nonumber\\
&=2\sqrt{\lambda_0 \lambda_1}\nonumber\\
&=\sqrt{2(1-\T{\rho_A^2})}\nonumber\\
&=\mathcal{C}(\ket{\phi}),
\label{negativityC}
\end{align}
where $\rho_A = \T_B (\ket{\phi}\bra{\phi})$. In other words,
negativity is equivalent to concurrence for any pure state with
Schmidt rank 2, and consequently it follows that for any 2-qubit
mixed state $\rho_{AB}=\sum_{i}p_{i}\ket{\phi_i}\bra{\phi_i}$,
\begin{align}
\mathcal{N}_{c}(\rho_{AB})=& \min \sum_{i}p_{i}\mathcal{N} (\ket{\phi_i})\nonumber\\
=&\min \sum_{i}p_{i}\mathcal{C} (\ket{\phi_i})\nonumber\\
=&\mathcal{C}(\rho_{AB}),
\label{CRENC}
\end{align}
where the minima are taken over all pure state decompositions of $\rho_{AB}$.

Similar to the duality between concurrence and CoA, we can also
define a dual to CREN, namely CRENoA,
by taking the maximum value of average negativity over all
possible pure state decomposition. Furthermore, for a two-qubit
system, we have
\begin{align}
\mathcal{N}_{c}^{a}(\rho_{AB})=& \max \sum_{i}p_{i}\mathcal{N} (\ket{\phi_i})\nonumber\\
=&\max \sum_{i}p_{i}\mathcal{C} (\ket{\phi_i})\nonumber\\
=&\mathcal{C}^{a}(\rho_{AB}),
\label{CRENCoA}
\end{align}
where maxima are taken over all pure state decompositions of
$\rho_{AB}$ and $\mathcal{N}_{c}^{a}(\rho_{AB})$ is the CRENoA of
$\rho_{AB}$.

From the analysis of CREN and CRENoA, we can see that CREN can be
considered as a generalized version of concurrence from 2-qubit
systems. Thus, having the monotonicity and separability criteria of
CREN, it is natural to investigate MoE
in terms of CREN for multi-qubit systems and possible higher-dimensional quantum systems.
%%%%%%%%%%%%%%%%%%%%%%%%%%%%%%%%%%%%%%%%%%%%%%%%%%%%%%%%%%%%%%%%%%%%%
%%                                                                %%%
%%                  CREN monogamy of entanglement                 %%%
%%                                                                %%%
%%%%%%%%%%%%%%%%%%%%%%%%%%%%%%%%%%%%%%%%%%%%%%%%%%%%%%%%%%%%%%%%%%%%%
\section{CREN Monogamy of Entanglement}
\label{Sec: CREN Monogamy of Entanglement}

In three-qubit systems, Coffman, Kundu and Wootters (CKW)~\cite{ckw}
first introduced a monogamy inequality in terms of concurrence, as
\begin{equation}
\mathcal{C}_{A(BC)}^2\ge\mathcal{C}_{AB}^2+\mathcal{C}_{AC}^2,
\label{CKW}
\end{equation}
where $\mathcal{C}_{A(BC)}=\mathcal{C}(\ket{\psi}_{A(BC)})$ is the concurrence of a 3-qubit state $\ket{\psi}_{A(BC)}$ for
a bipartite cut of subsystems between $A$ and $BC$ and $\mathcal{C}_{AB}=\mathcal{C}(\rho_{AB})$.
Similarly, its dual inequality in terms of CoA,
%which can be interpreted as  the maximum of the average of concurrence taken over all possible pure state decompositions
%of the mixed state,
\begin{equation}
\mathcal{C}_{A(BC)}^2\le(\mathcal{C}_{AB}^a)^2+(\mathcal{C}_{AC}^a)^2,
\label{dual}
\end{equation}
has been shown in~\cite{gms}.
Later, the CKW inequality has been generalized into $n$-qubit systems~\cite{ov},
and its dual inequality for $n$-qubit systems has also been introduced~\cite{gbs}.

However, a quantum state in a $3 \otimes 3\otimes 3$ quantum system was found that violates
the CKW inequality~\cite{ou}, and recently another counterexample was found in
a $3 \otimes 2\otimes 2$ quantum system~\cite{ks}; therefore
the CKW inequality only holds for multi-qubit systems, and even a tiny extension in any of the subsystems leads
to a violation.

In this section, we show that all the monogamy inequalities for qubits using concurrence
can be rephrased by CREN, and this CREN monogamy inequality is still true for
the counterexamples in~\cite{ou, ks}.

\subsection{Monogamy Inequalities for $n$-qubit systems in terms of CREN}
\label{Subsec: Monogamy inequalities for $n$-qubit systems in terms of CREN}

For any pure state $\ket{\psi}_{A_1 \cdots A_n}$ in an $n$-qubit
system $A_1 \otimes \cdots \otimes A_n$ where $A_i \cong
\mathbb{C}^2$ for $i=1,\ldots,n$, a generalization of the CKW
inequality,
\begin{equation}
\mathcal{C}_{A_1 (A_2 \cdots A_n)}^2  \geq  \mathcal{C}_{A_1 A_2}^2 +\cdots+\mathcal{C}_{A_1 A_n}^2,
\label{nCKW}
\end{equation}
was conjectured~\cite{ckw} and proved~\cite{ov}.
 Another inequality, which can be
considered to be dual to Eq.~(\ref{nCKW}) was
also introduced in~\cite{gbs},
\begin{equation}
\mathcal{C}_{A_1 (A_2 \cdots A_n)}^2  \leq  (\mathcal{C}^a_{A_1 A_2})^2 +\cdots+(\mathcal{C}^a_{A_1 A_n})^2 .
\label{ndual}
\end{equation}

Now, let us consider these inequalities in terms of CREN. First,
note that any $n$-qubit pure state $\ket{\psi}_{A_1 \cdots A_n}$
can have a Schmidt decomposition with at most two non-zero Schmidt
coefficients with respect to the bipartite cut between $A_1$ and
the others. Thus, by Eq.~(\ref{negativityC}), we have
\begin{equation}
\mathcal{C}_{A_1 (A_2 \cdots A_n)} = {\mathcal{N}_{c}}_{A_1 (A_2 \cdots A_n)}.
\label{ncrenC}
\end{equation}
Furthermore, for any reduced density matrix $\rho_{A_i A_j}$ of $\ket{\psi}_{A_1 \cdots A_n}$ onto two-qubit
subsystems $A_i \otimes A_j$,
%with $i,~j \in \{1,\ldots, n \},~i\neq j$,
it is a two-qubit mixed state; therefore, by Eqs.~(\ref{CRENC}) and (\ref{CRENCoA}), we have
\begin{equation}
\mathcal{C}_{A_i A_j} = {\mathcal{N}_c}_{A_i A_j },~\mathcal{C}^a_{A_i A_j} = {\mathcal{N}_{c}}^a_{A_i A_j },
\label{ncrenC2}
\end{equation}
for $i, j \in \{1,\cdots, n \},~ i\neq j$.

Thus, we have the following theorem.
\begin{Thm}
For any $n$-qubit pure state $\ket{\psi}_{A_1 \cdots A_n}$,
\begin{equation}
{{{\mathcal{N}}_c}_{A_1 (A_2 \cdots A_n)}}^2  \geq
{{\mathcal{N}_c}_{A_1 A_2}}^2 +\cdots+{{\mathcal{N}_c}_{A_1
A_n}}^2, \label{nineq cren}
\end{equation}
and
\begin{equation}
{{{\mathcal{N}}_c}_{A_1 (A_2 \cdots A_n)}}^2  \leq
({\mathcal{N}_c}^a_{A_1 A_2})^2 +\cdots+({\mathcal{N}_c}^a_{A_1
A_n})^2, \label{nineq cren a}
\end{equation}
where ${{\mathcal{N}}_c}_{A_1 (A_2 \cdots
A_n)}={\mathcal{N}}(\ket{\psi}_{A_1(A_2 \cdots A_n)})$,
% with respect to the bipartite cut between $A_1$ and the others,
${\mathcal{N}_c}_{A_1 A_i}={\mathcal{N}_c}(\rho_{A_1 A_i})$ and
${\mathcal{N}_c}^a_{A_1 A_i}={\mathcal{N}_c}^a(\rho_{A_1 A_i})$
for $i=2, \ldots, n$.
\label{CRENmono}
\end{Thm}
\begin{proof}
It is a direct consequence from the overlap of CREN and concurrence
in Eqs.~(\ref{ncrenC}) and (\ref{ncrenC2}), as well as  the monogamy inequalities in
Eqs.~(\ref{nCKW}) and (\ref{ndual}) by concurrence.
\end{proof}

%In other words, for every multi-qubit monogamy inequality so far,
%Eqs.~(\ref{nCKW}) and (\ref{ndual}) can be rephrased in terms
%of CREN in a exactly same way as  Eqs.~(\ref{nineq
%cren}) and (\ref{nineq cren a}).

In~\cite{of}, another monogamy inequality of entanglement for
three-qubit systems in terms of the original
negativity~\cite{VidalW} was proposed. For a three-qubit state
$\ket{\psi}_{ABC}$, it was shown that
\begin{equation}
{\mathcal{N}_{A(BC)}}^2\geq {\mathcal{N}_{AB}}^2+{\mathcal{N}_{AC}}^2,
\label{3nega}
\end{equation}
where ${\mathcal{N}_{AB}}^2=\|\rho_{AB}^{T_B}\|_1-1 $ and
${\mathcal{N}_{AC}}^2=\|\rho_{AC}^{T_C}\|_1-1 $ are the original
negativities of $\rho_{AB}$ and $\rho_{AC}$ respectively.

Due to the convexity of the original negativity, we can easily see
that CREN is always an upper bound of the original negativity. In
other words, for any bipartite mixed state $\rho_{AB}$,
\begin{equation}
\mathcal{N}_c(\rho_{AB}) \geq \mathcal{N}(\rho_{AB}).
\label{upperbound}
\end{equation}
From Theorem~\ref{CRENmono} together with Eq.~(\ref{upperbound}), we have the following corollary
which encapsulates the result of Eq.~(\ref{3nega}).
\begin{Cor}
For any $n$-qubit pure state $\ket{\psi}_{A_1 \cdots A_n}$,
\begin{equation}
{{\mathcal{N}}_{A_1 (A_2 \cdots A_n)}}^2
%\geq & {{\mathcal{N}_c}_{A_1 A_2}}^2 +\cdots+{{\mathcal{N}_c}_{A_1 A_n}}^2\nonumber\\
\geq  {{\mathcal{N}}_{A_1 A_2}}^2 +\cdots+{{\mathcal{N}}_{A_1
A_n}}^2. \label{original n ineq}
\end{equation}
\end{Cor}

Thus, besides concurrence, CREN is another good entanglement
measure in multi-qubit systems for MoE.

\subsection{CREN vs Concurrence-based Monogamy Relations}
\label{Subsec: CREN Monogamy Relations in the Counterexamples of CKW inequality}

Two counterexamples in~\cite{ou, ks} are, in fact, all known
counterexamples showing the violation of the CKW inequality in
higher-dimensional quantum systems. Here we show that they still
have a monogamy relation in terms of CREN.

\begin{counterexample}\label{Ex:1}
(Ou~\cite{ou})

Let us consider a pure state $\ket{\psi}$ in $3 \otimes 3\otimes
3$ quantum systems such that
\begin{align}  \label{f}
|\psi\rangle_{ABC}=\frac{1}{\sqrt{6}}(&|123\rangle-|132\rangle+|231\rangle\nonumber\\
&-|213\rangle+|312\rangle-|321\rangle).
\end{align}

Since $|\psi\rangle_{ABC}$ is pure, it is easy to check
$C^{2}_{A(BC)}=\frac{4}{3}$. For mixed states $\rho_{AB}$ and
$\rho_{AC}$, it was shown that any pure state in any pure state
ensemble has the same constant value, 1, as its concurrence, which
implies $C^{2}_{AB}=C^{2}_{AC}=1$. Therefore we have
\begin{equation}  \label{i}
C_{AB}^{2}+C_{AC}^{2}=2\geq \frac{4}{3}=C_{A(BC)}^{2},
\end{equation}
which is a violation of the CKW inequality in higher-dimensional
quantum systems.

Now, let us consider the case of using CREN as the entanglement
measure for the state in Eq.~(\ref{f}).

Since $|\psi\rangle_{ABC}$ is pure, it can be easily checked that
\begin{equation}
\mathcal{N}_{A(BC)}={\mathcal{N}_c}_{A(BC)}=(\T{ \sqrt{\rho_A}})^2 -1 =2.
\label{counterN123}
\end{equation}

For ${\mathcal{N}_c}_{AB}$, let us consider $\rho_{AB}$ whose
spectral decomposition is,
\begin{equation}
\rho_{AB}=\frac{1}{3}\left(|x\rangle_{AB}\langle x|+|y\rangle_{AB}\langle
y|+|z\rangle_{AB}\langle z|\right),
\label{specrhoAB}
\end{equation}
where
\begin{align}
|x\rangle_{AB} =& \frac{1}{\sqrt{2}}\left(|23\rangle-|32\rangle\right),\nonumber\\
|y\rangle_{AB} =& \frac{1}{\sqrt{2}}\left(|31\rangle-|13\rangle\right),\nonumber\\
|z\rangle_{AB} =& \frac{1}{\sqrt{2}}\left(|12\rangle-|21\rangle\right).
\end{align}

By a straightforward calculation, it can be shown that for arbitrary
pure states $|\phi\rangle_{AB}=c_{1}|x%
\rangle_{AB}+c_{2}|y\rangle_{AB} + c_{3}|z\rangle_{AB}$ with $%
|c_{1}|^{2}+|c_{2}|^{2}+|c_{3}|^{2}=1$, their reduced density matrix $%
\rho_{A} = \T_{B}|\phi\rangle_{AB}\langle\phi|$ has the same
spectrum $\{\frac{1}{2}, \frac{1}{2}, 0\}$~\cite{vidal}. Thus, we have
\begin{equation}
\mathcal{N}(\ket{\phi}_{AB})=(\T{ \sqrt{\rho_A}})^2 -1 =1,
\label{same}
\end{equation}
for any $\ket{\phi}_{AB}$ that is a superposition of
$\ket{x}_{AB}$, $\ket{y}_{AB}$ and $\ket{z}_{AB}$. By the
Hughston-Jozsa-Wootters (HJW) theorem~\cite{HJW}, any pure state
in any pure state ensemble of $\rho_{AB}$ can be realized as a
superposition of $\ket{x}_{AB}$, $\ket{y}_{AB}$ and
$\ket{z}_{AB}$ thus we have
\begin{align}
{\mathcal{N}_c}(\rho_{AB})=& \min_{\sum_k p_k \ket{\phi}_k\bra{\phi_k}=\rho_{AB}} \sum_k p_k \mathcal{N} (\ket{\phi}_k)\nonumber\\
=& \frac{1}{3} \left(\mathcal{N} (\ket{x}_{AB})+\mathcal{N} (\ket{y}_{AB})+\mathcal{N} (\ket{z}_{AB})\right)\nonumber\\
=& 1.
\label{counterN12}
\end{align}
Since Eq.~(\ref{f}) is asymmetric, we also have a similar result for $\rho_{AC}$, which is
\begin{align}
{\mathcal{N}_c}(\rho_{AC})=& \min_{\sum_k p_k \ket{\phi}_k\bra{\phi_k}=\rho_{AC}} \sum_k p_k \mathcal{N} (\ket{\phi}_k)\nonumber\\
=& \frac{1}{3} \left(\mathcal{N} (\ket{x}_{AC})+\mathcal{N} (\ket{y}_{AC})+\mathcal{N} (\ket{z}_{AC})\right)\nonumber\\
=& 1.
\label{counterN13}
\end{align}
Now, from Eq.~(\ref{counterN123}) together with  Eqs.~(\ref{counterN12}) and (\ref{counterN13}), we have
\begin{equation}
{{\mathcal{N}_c}_{A(BC)}}^2 = 4 \geq  1+1={{\mathcal{N}_c}_{AB}}^2+{{\mathcal{N}_c}_{AC}}^2.
\label{nmono1}
\end{equation}
In other words, even though the state $\ket{\psi}$ in Eq.~(\ref{f})
is a counterexample of the CKW inequality in three-qutrit
systems in terms of concurrence, it still shows a monogamy
property in terms of CREN.
\end{counterexample}

\begin{counterexample}\label{Ex2}
(Kim and Sanders~\cite{ks})

Let us consider a pure state $\ket{\psi}$ in $3 \otimes 2 \otimes 2$ quantum systems such that
\begin{equation}
\ket{\psi}_{ABC} = \frac{1}{\sqrt{6}}(\sqrt{2}\ket{010}+\sqrt{2}\ket{101}+\ket{200}+\ket{211}).
\label{count2}
\end{equation}
It can be easily seen that $\mathcal {C}^2_{A(BC)}=\frac{12}{9}$ whereas $\mathcal {C}^2_{AB}=\mathcal {C}^2_{AC}=\frac{8}{9}$, which implies the violation of the CKW inequality. However, by using a similar method to the previous example, we can have ${{\mathcal{N}_c}_{A(BC)}}^2 = 4$ whereas
${{\mathcal{N}_c}_{AB}}^2 = {{\mathcal{N}_c}_{AB}}^2=\frac{8}{9}$, which implies the example in Eq.~(\ref{count2}) also shows a monogamy property in terms of CREN.
\end{counterexample}

Thus, CREN is a powerful alternative for MoE in multipartite higher-dimensional quantum systems
without any trivial counterexample so far.

%%%%%%%%%%%%%%%%%%%%%%%%%%%%%%%%%%%%%%%%%%%%%%%%%%%%%%%%%%%%%%%%%%%%%%%%%%
%                                                                        %
%                 Partially Coherent Superposition                       %
%                                                                        %
%%%%%%%%%%%%%%%%%%%%%%%%%%%%%%%%%%%%%%%%%%%%%%%%%%%%%%%%%%%%%%%%%%%%%%%%%%
\section{Partially Coherent Superposition of an $n$-qudit Generalized W-class state
and $\ket 0^{\otimes n}$}\label{Sec: Partially Coherent Superposition}

Three-qubit systems can have two inequivalent classes of genuine
tripartite entangled states by the CKW inequality~\cite{DVC}.
One of them is the Greenberger-Horne-Zeilinger (GHZ) class~\cite{GHZ}
and the other one is the W class~\cite{DVC}. These two classes show
extreme differences in terms of the CKW and its dual inequalities:
The CKW and its dual inequalities are saturated by W-class states,
whereas the terms for reduced density matrices in the inequalities
always vanish for GHZ-class states. Since the saturation of the
CKW inequality by W-class states can be interpreted as a genuine
tripartite entanglement with a complete characterization by means
of its partial entanglements, W-class states here are especially
interesting.

It was shown that there also exists a class of states in $n$-qudit
systems which saturate a monogamy relation~\cite{ks}.
By using concurrence as the entanglement
measure, the monogamy inequalities are shown to be saturated by
incoherent superpositions of a generalized $n$-qudit W-class state~\cite{ks}
and the vacuum, $\ket{0}^{\otimes n}$.

In this section, we propose a class of quantum states
in $n$-qudit systems consisting of partially coherent
superpositions of a generalized W-class state and the
vacuum, and show that they have the saturation of the monogamy relation
in terms of CREN. This saturation is also true for an arbitrary partition
of the set of subsystems, and it is not even affected by the degree of coherency.

%We believe this can be the first step for
%the research of general monogamy relation of entanglement
%in higher dimensional quantum systems in terms of CREN.

Let us reprise the definition of an $n$-qudit generalized W-class
state~\cite{ks},
\begin{align}
\left|W_n^d \right\rangle_{A_1\cdots A_n}=\sum_{i=1}^{d-1}(a_{1i}{\ket {i0\cdots 0}} +&a_{2i}{\ket {0i\cdots 0}}\nonumber\\
+&\cdots +a_{ni}{\ket {00\cdots 0i}}),\nonumber\\
\sum_{i=1}^{d-1}(|a_{1i}|^2+|a_{2i}|^2+\cdots +|a_{ni}|^2&)=1,
\label{generalized W state}
\end{align}
which is a coherent superposition of all $n$-qudit
product states with Hamming weight one.

A partially coherent superposition of a generalized W-class state
and $\ket{0}^{\otimes n}$ is given as
\begin{align}
\rho_{A_1\cdots A_n}=&p\left|W_n^d \right\rangle \left\langle W_n^d \right|+(1-p)\ket 0^{\otimes n}\bra 0^{\otimes n}\nonumber\\
&+\lambda \sqrt{p(1-p)}(|\left|W_n^d \right\rangle\bra 0^{\otimes n}+\ket 0^{\otimes n}\left\langle W_n^d \right|),
\label{partially coherent density matrix}
\end{align}
where $\lambda$ is the degree of coherence with $0 \leq \lambda
\leq 1$.
For the case that $\lambda =1$, Eq.~(\ref{partially
coherent density matrix}) becomes a coherent superposition of a
generalized W-class state and $\ket{0}^{\otimes n}$, and it is an
incoherent superposition, or a mixture when $\lambda =0$.
In other words, Eq.~(\ref{partially coherent density matrix}) is an $n$-qudit state where
the product state of Hamming weight zero is in a partially coherent superposition with all
the states of Hamming weight one.

The state in Eq.~(\ref{partially coherent density matrix}) can also be
interpreted by means of {\em decoherence}. In other words,
Eq.~(\ref{partially coherent density matrix}) can be considered as
the resulting state from a coherent superposition of a generalized
W-class state and $\ket{0}^{\otimes n}$,
\begin{equation}
\ket \psi _{A_1,\cdots A_n}=\sqrt{p} \left|W_n^d\right\rangle+\sqrt{1-p}\ket 0^{\otimes n},
\label{superposition}
\end{equation}
after some decoherence process so-called {\em phase damping}~\cite{nc},
which can be represented as
\begin{align}
\rho_{A_1\cdots A_n}&=\Lambda(\ket \psi \bra{\psi})\nonumber\\
&=E_{0}\ket{\psi}\bra{\psi}E_{0}^{\dag}+E_{1}\ket{\psi}\bra{\psi}E_{1}^{\dag}+E_{2}\ket{\psi}\bra{\psi}E_{2}^{\dag},
\label{Kraus}
\end{align}
with Kraus operators $E_{0}=\sqrt{\lambda}I$,
$E_{1}=\sqrt{1-\lambda}(I-\ket{0}\bra{0})$ and
$E_{2}=\sqrt{1-\lambda}\ket{0}\bra{0}$.

Now, we will see that the monogamy relation of the state in
Eq.~(\ref{partially coherent density matrix}) in terms of CREN is
saturated with respect to any arbitrary partition of the set of
subsystems. Furthermore, the entanglements, measured by CREN, of
the state in Eq.~(\ref{partially coherent density matrix}) and its
reduced density matrix onto any subsystem with respect to any
bipartite cut are not affected by the degree of coherency
$\lambda$.

First, let us consider the CREN of $\rho_{A_1\cdots A_n}$ in
Eq.~(\ref{partially coherent density matrix}) with respect to the
bipartite cut between $A_1$ and the others.
The state in Eq.~(\ref{partially coherent density matrix}) has a
pure state decomposition as
\begin{widetext}
\begin{align}
\rho_{A_1\cdots A_n}=&(\sqrt{p}\left|W_n^d\right\rangle+\lambda \sqrt{1-p}\ket{0}^{\otimes n})(\sqrt{p} \left\langle W_n^d \right|+\lambda \sqrt{1-p}\bra{0}^{\otimes n})\nonumber\\
&+(\sqrt{(1-p)(1-\lambda^2)}\bra 0^{\otimes n})(\sqrt{(1-p)(1-\lambda^2)}\ket 0^{\otimes n}).
%\nonumber\\
%=& |\tilde{\psi}_1\rangle \langle\tilde{\psi}_1|+|\tilde{\psi}_2\rangle \langle\tilde{\psi}_2|,\nonumber\\
\end{align}
\end{widetext}
Now, let
\begin{align}
|\tilde{\psi}_1\rangle =&\sqrt{p}\left|W_n^d\right\rangle+\lambda
\sqrt{1-p}\ket 0^{\otimes n},\nonumber\\
|\tilde{\psi}_2\rangle =&\sqrt{(1-p)(1-\lambda^2)}\ket 0^{\otimes
n}
\label{spec1}
\end{align}
 be two unnnormalized states in an $n$-qudit system.
Then, by the HJW theorem~\cite{HJW}, any other pure state decomposition of
$\rho_{A_1(A_2 \cdots A_n)}=\sum_{i=1}^{r}|\tilde{\phi}_i \rangle
\langle \tilde{\phi}_i |$ of size $r$ can be realized by the
choice of an $r\times r$ unitary matrix $(u_{ij})$ such that
$|\tilde{\phi}_i \rangle=u_{i1}|\tilde{\psi}_1 \rangle+u_{i2}
|\tilde{\psi}_2 \rangle$. In other words, with the normalization
of $|\tilde{\phi}_i \rangle=\sqrt{p_i}|{\phi}_i \rangle$, we can
consider an arbitrary pure state decomposition of $\rho_{A_1(A_2
\cdots A_n)}=\sum_{i=1}^{r}p_i |{\phi}_i \rangle \langle {\phi}_i
|$ with arbitrary size $r$.

By using the method introduced in~\cite{ks}, we can directly
evaluate the average negativity of the pure states $\ket{\phi_i}$
for an arbitrary pure state decomposition of $\rho_{A_1(A_2 \cdots
A_n)}$. After tedious but straightforward calculations, it can be
shown that the average negativity is independent from the choice
of a unitary matrix $(u_{ij})$, which is,
\begin{equation}
\sum_{i}p_i \mathcal{N}(|{\phi}_{i} \rangle)=2p\sqrt{\mathcal{A}(1-\mathcal{A})},
\label{ave 1...n}
\end{equation}
where $\mathcal{A}=1-\sum_{j=1}^{d-1}|a_{1j}|^2$.

Thus, by the definition of CREN, we have
\begin{align}
{\mathcal{N}_c}(\rho_{A_1(A_2 \cdots A_n)})=&\min \sum_{i}p_i \mathcal{N}(|{\phi}_{i} \rangle)\nonumber\\
=&2p\sqrt{\mathcal{A}(1-\mathcal{A})},
\label{CRENA1..An}
\end{align}
where the minimum is taken over all possible pure state
decompositions of $\rho_{A_1(A_2 \cdots A_n)}=\sum_{i}p_i
|{\phi}_i \rangle \langle {\phi}_i |$.

Furthermore, it can be seen from Eq.~(\ref{ave 1...n}) that this average value is also invariant under the degree
of coherency $\lambda$. In other words, no matter how much amount of
decoherence in Eq.~(\ref{Kraus}) happens to the state in  Eq.~(\ref{superposition}), its entanglement is preserved.

Now, for ${\mathcal{N}_c}_{A_1A_i}$ with $i=2,..., n$, let us
first consider the case when $i=2$, whereas all the other cases are
analogously following. By tracing over all subsystems except $A_1$
and $A_2$ from $\rho_{A_1 \cdots A_n}$, we get
\begin{widetext}
\begin{align}
\rho_{A_1A_2}=&p\sum_{i,j=1}^{d-1}[a_{1i}a^*_{1j}\ket {i0}\bra {j0}+a_{1i}a^*_{2j}\ket {i0}\bra {0j}+a_{2i}a^*_{1j}\ket {0i}\bra {j0}
%\nonumber\\
+a_{2i}a^*_{2j}\ket {0i}\bra {0j}]\nonumber\\
&+(\mathcal{A}_{2}+1-p)\ket {00}\bra {00}\nonumber\\
&+\lambda \sqrt{p(1-p)}\sum_{k=1}^{d-1}[(a_{1k}\ket {k0}+a_{2k}\ket {0k})\bra {00}
%\nonumber\\
+a^*_{1k}\ket {00}(\bra {k0}+a^*_{2k}\bra {0k})],
%\nonumber\\
%=& |\tilde{\psi}_1\rangle \langle\tilde{\psi}_1|+|\tilde{\psi}_2\rangle \langle\tilde{\psi}_2|,
\label{partially coherent A1 A2 matrix}
\end{align}
\end{widetext}
with $\mathcal{A}_{2}=1-\sum_{j=1}^{d-1}(|a_{1j}|^2+|a_{2j}|^2)$.

Let us consider two unnormalized states
\begin{align}
|\tilde{\psi}_1\rangle=&\sqrt{p}\sum_{i=1}^{d-1}(a_{1i}\ket {i0}+a_{2i}\ket
{0i})+\lambda\sqrt{1-p}\ket {00},\nonumber\\
|\tilde{\psi}_2\rangle=&\sqrt{\mathcal{A}_{2}+(1-p)(1-\lambda^2)}\ket {00},
\end{align}
then we have
\begin{equation}
\rho_{A_1A_2}=|\tilde{\psi}_1\rangle \langle \tilde{\psi}_1|+|\tilde{\psi}_2\rangle \langle \tilde{\psi}_2|.
\end{equation}
Thus all possible pure states in an arbitrary pure state
decomposition of $\rho_{A_1 A_2}$ of size $r$ can be realized as a
linear combination of $|\tilde{\psi}_1\rangle$ and
$|\tilde{\psi}_2\rangle$ by choosing an $r \times r$ unitary
matrix. Again, by using a similar method to the case of
$\rho_{A_1 \cdots A_n}$, it can been shown that the average
negativity of $\rho_{A_1 A_2}$ is invariant under the choice of
pure state decomposition, which is,
\begin{equation}
{\mathcal{N}_c}_{A_1A_2}=2p\sqrt{(1-\mathcal{A})(\mathcal{A}-\mathcal{A}_{2})}.
\label{pcoherentN12}
\end{equation}
Furthermore, rather surprisingly, this average value is also
invariant under the degree of coherency. In other words, no matter
how much amount of decoherence in Eq.~(\ref{Kraus}) happens, it does not even affect the entanglement
between the subsystems $A_1$ and $A_2$.

Similarly, we can have,
\begin{equation}
{\mathcal{N}_c}_{A_1A_i}=2p\sqrt{(1-\mathcal{A})(\mathcal{A}-\mathcal{A}_{i})},~i=3,..., n
\label{pcoherentN1i}
\end{equation}
with $\mathcal{A}_{i}=1-\sum_{j=1}^{d-1}(|a_{1j}|^2+|a_{ij}|^2)$, and thus,
\begin{align}
\sum_{i=2}^{n}{\mathcal{N}_c}^2_{A_1A_i}=&\sum_{i=2}^{n}{4p^2(1-\mathcal{A})(\mathcal{A}-\mathcal{A}_{i})}\nonumber\\
=& 4p^2\mathcal{A}(1-\mathcal{A})\nonumber\\
=& {\mathcal{N}_c}^2_{A_1(A_2\cdots A_n)}.
\label{saturation}
\end{align}

In other words, we have obtained a saturation of the CREN monogamy relation for an $n$-qudit state
in Eq.~(\ref{partially coherent density matrix}), and this saturation does not depend on the choice of coherency $\lambda$.

For any arbitrary partition $P=\{P_1, \cdots, P_m \}$ of the set of subsystems, it was shown that an $n$-qudit
generalized W-class state can be also considered as an $m$-partite generalized W-class state~\cite{ks}, that is,
\begin{widetext}
\begin{align}
|W_n^d\rangle_{A_1\cdots A_n}=&\sum_{i=1}^{d-1}(a_{1i}{\ket {i\cdots 0}} +\cdots +a_{ni}{\ket {0\cdots i}})\nonumber\\
=&\sum_{i=1}^{d-1} |\tilde{x}_{1i}\rangle_{P_1}\otimes\cdots\otimes|\vec{0}\rangle_{P_m}
           +\cdots +|\vec{0}\rangle_{P_1}\otimes\cdots\otimes|\tilde{x}_{mi}\rangle_{P_m}\nonumber\\
=&\sum_{i=1}^{d-1} \sqrt{q_{1i}}\ket{i}_{P_1}\otimes\cdots\otimes\ket{0}_{P_m}
            +\cdots+\sqrt{q_{mi}} \ket{0}_{P_1}\otimes\cdots\otimes\ket{i}_{P_m}\nonumber\\
=&|W_m^d\rangle_{P_1\cdots P_m},
\label{n-mequivalence}
\end{align}
\end{widetext}
where
\begin{align}
|\tilde{x}_{si}\rangle_{P_s}=a_{(n_1+\cdots+n_{s-1}+1)i}&\ket{i\cdots 0}_{P_s}+\nonumber\\
&\cdots+a_{(n_1+\cdots+n_s) i}\ket{0\cdots i}_{P_s}
\end{align}
and
\begin{equation}
\sqrt{q_{si}}\ket{x_{si}}_{P_s} = |\tilde{x}_{si}\rangle_{P_s},~|\vec{0}\rangle_{P_s} = \ket{0\cdots0}_{P_s}
\end{equation}
with renaming $\ket{x_{si}}_{P_s}=\ket{i}_{P_s}$ and $|\vec{0}\rangle_{P_s}= \ket{0}_{P_s}$ for $s \in \{1, \ldots,m \}$.

Therefore Eq.~(\ref{partially coherent density matrix}) can also be considered to be a partially coherent superposition of an $m$-partite generalized W-class state and the vacuum, $\ket{0}_{P_1 \cdots P_m}$, and thus the result in (\ref{saturation}) is also true for any arbitrary partition of the set of subsystems.

Not only for the case of multi-qubit systems and the counterexamples in Sec.~\ref{Sec: CREN Monogamy of Entanglement},
CREN also shows a strong monogamy relation of entanglement for a class of $n$-qudit states in a partially coherent mixture
of a generalized W-class state and the vacuum. Thus, the CREN version of MoE is a strong
conjecture for qudit systems with no obvious counterexamples.
%%%%%%%%%%%%%%%%%%%%%%%%%%%%%%%%%%%%%%%%%%%%%%%%%%%%%%%%%%%%%%%%%%%%%%%%%%
%                                                                        %
%                           Conclusions                                  %
%                                                                        %
%%%%%%%%%%%%%%%%%%%%%%%%%%%%%%%%%%%%%%%%%%%%%%%%%%%%%%%%%%%%%%%%%%%%%%%%%%
\section{Conclusions}\label{Sec: Conclusion}

The study of higher-dimensional quantum systems is, undoubtedly,
important and even necessary to quantum information science for
various kind of reasons. First, qudits for $d>2$ are preferred in
some physical systems such as in quantum key distribution where
the use of qutrits increases coding density and provide stronger
security compared to qubits~\cite{gjvw}. In fault-tolerant quantum
computation as well as on quantum error-correcting codes (QECCs),
many studies are concentrated on the case of binary QECCs in a
two-dimensional Hilbert space, whereas generalizations of proofs
are often non-trivial when $d>2$.

However, as both qubit and qudit systems occur in the natural
world, there is no reason to assume that a theoretical result
should hold solely for two-dimensional systems. If an important
result (e.g. monogamy of entanglement) is shown to be true for the
case $d=2$, then this would suggest that a lot of effort should be
directed towards qudit systems, as the case for $d > 2$ could be
fundamentally different from the case $d=2$. For example a recent
result~\cite{cccz} shows that for subsystem stabilizer codes in
$d$ dimensional Hilbert space, a universal set of transversal
gates cannot exist for even one encoded qudit, for any dimension
$d$, which is known as {\em no-go theorem} for the universal set
of transversal gates in QECC.

The extension of the multipartite entanglement analysis,
especially the monogamy relation from qubit to qudit case is far
more than trivial.
The entanglement properties in higher-dimensional systems are hardly known so far,
and thus any fundamental step of the challenges to the richness of entanglement
studies for system of higher-dimensions and multipartite systems
would be fruitful and even necessary to understand the whole
picture of quantum entanglement.

In this paper, we have proposed CREN as a powerful alternative for
MoE in higher-dimensional quantum systems.
We have shown that any monogamy inequality of entanglement for multi-qubit systems can be
rephrased in terms of CREN. Furthermore we have pointed out the possibility
of CREN MoE in higher-dimensional quantum systems by showing that all the
counterexamples for the CKW inequality so far in higher-dimensional
quantum systems still have a monogamy inequality in terms of CREN,
as well as no trivial counterexamples for CREN MoE so far.
This task is one of the key challenges in finding a bipartite entanglement measure that
meets our three criteria for qubits and for higher-dimensional systems.

For the studies of CREN MoE in higher-dimensional
quantum systems, we have proposed a class of quantum
states in $n$-qudit systems that are in a partially coherent
superpositions of a generalized W-class state and the vacuum.
The CREN monogamy relation for the proposed
class has been shown to be true and it also holds with respect to
any arbitrary partition of the subsystems.

Thus CREN is a good candidate for the general monogamy relation of multipartite entanglement,
and it shows a strong evidence of its possibility even for the case of mixed states in higher-dimensional
systems.
We believe that the analysis of CREN MoE derived here will give a full and rich
reference for the study of MoE in higher-dimensional quantum systems, which is one of the
most important and necessary tasks in study of quantum entanglement.

%Dealing with pure states, as well as their entanglement, are
%extremely subtle and hard in real world, and this is mainly due to
%the decoherence, which tends to corrupt the states. However the
%proposed class of the states here has been shown that their
%entanglements (measured by CREN) are not affected by the degree of
%coherency in the superposition.
\section*{Acknowledgments}
This work is supported by the Korea Research Foundation Grant funded by the Korean Government (MOEHRD) (KRF-2007-357-C00008) and
Alberta's informatics Circle of Research Excellence (iCORE).
BCS is supported by CIFAR Associateship and MITACS.

%%%%%%%%%%%%%%%%%%%%%%%%%%%%%%%%%%%%%%%%%%%%%%%%%%%%%%%%%%%%%%%%%%%%%%%%


\begin{thebibliography}{1}

\bibitem{tele}
C. H. Bennett, G. Brassard, C. Crepeau, R. Jozsa, A. Peres and W. K. Wootters,
Phys. Rev. Lett. {\bf 70}, 1895 (1993).

\bibitem{qkd1}
C. Bennett and G. Brassard, {\em in Proceedings of IEEE
International Conference on Computers, Systems, and Signal
Processing }(IEEE Press, New York, Bangalore, India, 1984), p.
175-179.

\bibitem{qkd2}
C. H. Bennett, Physical Review Letters {\bf 68}, 3121 (1992).

\bibitem{ckw}
V. Coffman, J. Kundu and W. K. Wootters,
Phys. Rev. A {\bf 61}, 052306 (2000).

\bibitem{rg}
J. M. Renes and M. Grassl,
Phys. Rev. A {\bf74}, 022317 (2006).

\bibitem{d}
P. A. M. Dirac,
Proc. Roy. Soc. Lond. A {\bf 112}, 661-677 (1926).

\bibitem{h}
W. Heisenberg,
Z. Physik {\bf 49}, 619-636 (1928).


\bibitem{anti}
A. J. Coleman and V. I. Yukalov,
Lecture Notes in Chemistry Vol. {\bf 72} (Springer-Verlag, Berlin, 2000).


\bibitem{ov}
T. Osborne and F. Verstraete,
Phys. Rev. Lett. {\bf 96}, 220503 (2006).

\bibitem{ou}
Y. Ou,
Phys. Rev. A {\bf 75}, 034305 (2007).

\bibitem{ks}
J. S. Kim and B. C. Sanders,
J. Phys. A {\bf 41}, 495301 (2008).

\bibitem{LCOK}
S. Lee, D. P. Chi, S. D. Oh and J. Kim,
Phys, Rev. A {\bf 68}, 62304 (2003).

\bibitem{ww}
W. K. Wootters,
Phys. Rev. Lett. {\bf 80}, 2245 (1998).

\bibitem{lve}
T. Laustsen, F. Verstraete and S. J. van Enk, Quantum Inf. Comput.
{\bf 3}, 64 (2003).

\bibitem{VidalW}
G.~Vidal and R. F.~Werner,
%{\em Computable measure of entanglement},
Phys. Rev. A {\bf 65}, 032314 (2002).

\bibitem{Peres} A.~Peres,
%{\em Separability criterion for density matrices},
Phys. Rev. Lett. {\bf 77}, 1413 (1996).

\bibitem{Horodeckis1} M.~Horodecki, P.~Horodecki and R.~Horodecki,
 %{\em Separability of mixed states: necessary and sufficient conditions},
Phys. Lett. A {\bf 223}, 1 (1996).

\bibitem{Negativity}
Vidal and Werner~\cite{VidalW} defined the negativity of a state $\rho$ as $\frac{\|\rho^{T_B}\|_1-1}{2}$,
which corresponds to the absolute value of the sum of negative eigenvalues of $\rho^{T_B}$.
Another definition of the negativity with a normalizing factor, $\frac{\|\rho^{T_B}\|_1-1}{d-1}$,
was also used for the states in a $d\otimes d'$ ($d\le d'$) quantum system~\cite{LCOK}.
To avoid this inconsistency,
we only use $\|\rho^{T_B}\|_1-1$ here for our definition of the negativity, and this is also for the coincidence
of negativity with concurrence in two-qubit systems.
In addition, if we are interested in the monogamy relations of entanglement in multipartite quantum system, it is more reasonable
to give the possiblity of having larger negativity values to the states in higher-dimensional quantum system, rather than nomalizing
it in terms of the dimensionality of the systems as in~\cite{LCOK} where it always has the value 1 for maximally entangled states.

\bibitem{Horodecki1} P.~Horodecki,
%{\em Separability criterion and inseparable mixed states with positive partial transposition},
Phys. Lett. A {\bf 232}, 333 (1997).

\bibitem{DCLB} W.~D\"{u}r, J.I.~Cirac, M.~Lewenstein and D.~Bru\ss,
%{\em Distillabillty and partial transposition in bipartite system},
Phys. Rev. A {\bf 61}, 062313 (2000).

\bibitem{Horodeckis2} M.~Horodecki, P.~Horodecki and R.~Horodecki,
 %{\em Mixed-state entanglement and distillation: is there a ``bound'' entanglement in nature?},
Phys. Rev. Lett. {\bf 80}, 5239 (1998).

\bibitem{unnomalization}
CREN was originally introduced ~\cite{LCOK} with normalization. Here, we are using the same terminology `CREN' due to the same concept of extension for mixed-state case, even though we omit the normalization factor.

\bibitem{gms}
G. Gour, D. Meyer and B. C. Sanders, Phys. Rev. A  {\bf 72},
042329 (2005).

\bibitem{gbs}
G. Gour, S. Bandyopadhay and B. C. Sanders, J. Math. Phys. {\bf
48}, 012108 (2007).

\bibitem{of}
Y. Ou and H. Fan,
Phys. Rev. A {\bf 75}, 062308 (2007).

\bibitem{vidal}
G. Vidal, W. D\"{u}r and J. I. Cirac,
Phys. Rev. Lett. {\bf 89}, 027901 (2002).

\bibitem{HJW}
L. P. Hughston, R. Jozsa and W. K. Wootters,
%{\em A complete classification of quantum ensembles having a given density matrix},
Phys. Lett. A {\bf 183}, 14 (1993).

\bibitem{DVC}
W.~D\"{u}r, G. Vidal and J. I. Cirac,
%{\em   Three qubits can be entangled in two inequivalent ways},
Phys. Rev. A {\bf 62}, 062314 (2000).

\bibitem{GHZ}
D. M. Greenberger, M. A. Horne and A. Zeilinger,
{\em Bell's Theorem, Quantum Theory, and Conceptions of the Universe},
edited by M. Kafatos (Kluwer, Dordrecht, 1989), p. 69.

\bibitem{nc}
M. A. Nielsen and I. L. Chuang, {\em Quantum Computation
and Quantum Information} (Cambridge University Press, Cambridge,
U.K., 2000).

\bibitem{gjvw}
S. Groblacher, T. Jennewein, A. Vaziri, G. Weihs and A. Zeilinger,
New J. Phys. {\bf 8}, 75 (2006).

\bibitem{cccz}
X. Chen, H. Chung, A. W. Cross, B. Zeng and I. L. Chuang,
Phys. Rev. A {\bf78}, 012353 (2008).
\end{thebibliography}
\end{document}